\documentclass[a4paper, 11pt]{article}
 \usepackage{layout}
 \usepackage[dvipdfmx, dvips]{graphicx}
 \usepackage{chngcntr}
\usepackage{mathpazo}

\usepackage{amsmath, amssymb,bm,}
\usepackage{booktabs}
\usepackage{textcomp} 
\usepackage{ascmac} %
\usepackage[onehalfspacing]{setspace}

 \newtheorem{theorem}{Theorem}
\newtheorem{definition}{Definition}
\newtheorem{assumption}{Assumption}
\newtheorem{proposition}{Proposition} 
\newtheorem{lemma}{Lemma}
\newtheorem{corollary}{Corollary}
\newenvironment{proof}{\noindent \textbf{Proof.}}{\hfill$\blacksquare$\\[3pt]}

\begin{document}

\title{A Model of Financial Market Control} 

\author{OHASHI Yoshihiro\\
College of Economics, Nihon University\\
1-3-2 Kanda misaki-cho, Chiyoda-ku, Tokyo, 101-8360 Japan\\
e-mail: ohashi.yoshihiro@nihon-u.ac.jp\\
https://sites.google.com/view/yohashi/}
 
\maketitle

\begin{abstract}
This study investigates the prevention of market manipulation
 using a price-impact model of financial market trading as a linear system.
 First, I define a trading game between speculators 
 such that they implement a manipulation trading strategy that exploits momentum traders.
Second, I identify market intervention by a controller (e.g., a central bank) with 
a control of the system. 
The main result shows that there is a control strategy that prevents 
market manipulation 
as a subgame perfect equilibrium outcome of the trading game.
On the equilibrium path, no intervention is realized.  
This study also characterizes the set of 
manipulation-proof linear pricing rules of the system.
The set is very restrictive if there is no control, while the presence 
of control drastically expands the set. 
  \end{abstract}

\\
\textit{Journal of Economic Literature} Classification Numbers: G12, D49\\

\section{Introduction} 
This study analyzes the prevention of \textit{market manipulation}
 by a crowd of speculators who exploit \textit{momentum traders}. 
 Momentum traders follow a price trend, while speculators intentionally 
 cause a price trend and ``ignite'' the momentum. 
 Then, speculators buy and sell a security and obtain a profit margin. 
 This research focuses 
 on the question of whether a market system  can 
 spontaneously prevent this manipulation.

 To address this question, I investigate a \textit{price-impact} model 
 of financial market trading. 
 Price-impact models have been investigated in the literature of price formation (e.g., Kyle, 1985)
  or optimal execution (e.g., Almgren and Chriss, 2000).
 The  novel feature of this study is that it models the market system as a 
 \textit{linear system} and introduces a \textit{control} of the system.

 The trading model is based on Kyle (1985) and Huberman and Stanzl (2004).
 There are three kinds of participants in the market system: speculators,
 momentum traders, and a controller (e.g., the central bank 
 or government).  
  In each period $n$, customers simultaneously place a market order. 
  The market price in period $n$ is determined by a \textit{pricing rule}, which is a mapping from an aggregate market order to a non-negative price. 
  Each customer can trade at the market price.

A {pricing rule} of the system is said to be \textit{viable} (i.e., manipulation-proof) if it prevents speculators from implementing speculative trading strategies in some solution concept. 
I characterize the sets of viable pricing rules in the Nash equilibrium (NE) and subgame perfect equilibrium (SPE), which I refer to as \textit{NE-viable} pricing rules and \textit{SPE-viable} pricing rules, respectively, with or without a controller.

  The model of momentum traders is the same as that of 
\textit{positive feedback traders} in De Long et al. (1990).\footnote{
Hong and Stein (1999) analyze a model introducing momentum traders.
 Jegadeesh and Titman (1993) document the significance of momentum strategies. 
More recent empirical studies report momentum effects. 
See, for example, Moskowitz et al. (2012) and Baltzer et al. (2019).
}
As in their model, I assume that momentum traders are unintelligent agents:  they automatically
 buy a security today whenever they observe a price gain yesterday.
The price gain leads them to purchase the security, which in turn leads to further price gains.
 This self-perpetuating behavior continuously raises market prices.
  Speculators exploit this property of momentum traders and can earn a  profit.
  For example, a speculator buys a unit of the security and sells it when the price rises; when 
  the speculator sells the unit, another speculator simultaneously buys the same unit to 
  cancel out the negative price impact of the sale.

 The central results of this study are as follows. 
 First,  I characterize the set of NE-viable pricing rules  
 and the set of SPE-viable pricing rules in the absence of controls. 
 To compare the result of Huberman and Stanzl (2004), I also characterize the set of viable pricing 
 rules \textit{without} momentum traders \textit{and} a controller, 
 which I simply refer to as the \textit{maximal set}. 
 I find that both the sets of NE-viable and SPE-viable pricing rules
  are very restrictive compared to the {maximal set} (Proposition \ref{p2}). 
 In  particular, the result shows that the ``Kyle- type'' pricing rule, that is,  
 $p_n=p_{n-1}+\lambda_n q_n$ (Kyle (1985)),  is not NE-viable (and hence, not SPE-viable)
  in the absence of controls. 
  Second, I characterize the sets of NE-viable and SPE-viable pricing rules 
  in the presence of controls.
   I find that the set of SPE-viable pricing rules is equal to the maximal set with 
   a suitable control strategy (Theorem \ref{positive}). 
On the equilibrium path, the control strategy does not place a market
order.

  These results show that 
   the market system without a controller cannot spontaneously prevent market manipulation, 
   unless the system uses 
    \textit{very restrictive} pricing rules; if we allow the use of any viable pricing rule, control 
    by a third party is necessary. 

   This result is a new finding on the viable pricing rules.
  Huberman and Stanzl (2004) show that the linear pricing rules are the key to viability.
  According to their result,  some linear pricing rules are sufficient to
prevent the market manipulation of my model 
  \textit{if} there are no momentum traders. 
  However, their result is not relevant when market prices show a trend. 
  The set of viable pricing rules in the environment of Huberman and Stanzl (2004)
  is the maximal set of my model.
  The main finding of this study is that the set remains viable in my environment if and only if the
   control is present.

  A related study, Ohashi (2018), analyzes the strategic trading model
between a single speculator and multiple dealers in the market with momentum
traders. He shows that rational dealers might use the Kyle-type
pricing rule to exploit momentum traders.

\section{The Model}\label{model}

  Assume that there is a single security and 
   an infinite trading period $n\in \{1,2,3\cdots\}=:\mathbb{N}$.
 There are three types of market participants in the market system: 
 speculators, momentum traders, and a controller.
 Speculators are risk neutral, and each of them lives for two periods: 
 speculator $n$ enters the market in period $n$,
 trades in periods $n$ and $n+1$, and exits the market before 
 period $n+2$. 
 Momentum traders are infinitely lived. 
 Their trading behavior is proportional to past price movements 
 (see Assumption \ref{A}).
 The controller is also infinitely lived. 
 This study investigates two cases: the absence and presence 
 of controls.

 The market system is described as a \textit{pricing rule}
  following Huberman and Stanzl (2004).
  At the beginning of each period $n$, 
 the market system bids a price quote, $\tilde {p}_n$.
   Each customer can observe the quote before submitting his or her order.
 For market order $q_n$, which is a quantity of the security,
the system determines the market price, $p_n$, as 
 \begin{equation}\label{im}
    p_n=\tilde{p}_n+P_n(q_n)
 \end{equation}
where $P_n(q_n)$ is a \textit{price-impact function}, 
which describes the  immediate price reaction to market order $q_n$.
 Each customer trades at $p_n$ in period $n$.
 If necessary, the market system clears the market by using its own inventory of the security. 
 After the trade, the system updates the quote as
 \begin{equation}\label{up}
  \tilde{p}_{n+1}=\tilde{p}_n+U_n(q_n)
 \end{equation}
and proceeds to the next period, where $U_n(q_n)$ is a 
\textit{price-update function} that captures 
  only trade's permanent price impact. 
Eqs.\,(\ref{im}) and (\ref{up}) yield
\begin{equation}\label{single}
p_{n+1}= p_{n}+U_n(q_n)-P_n(q_n)+P_{n+1}(q_{n+1})
\end{equation}
and
\[
p_n= p_0+\sum^{n-1}_{k=1}U_k(q_k)+P_n(q_n),
\] 
 where $p_0$ is the initial price in the market. 
I refer to $( U_n, P_n,)$ as a \textit{pricing rule} in period $n$.
 This model ignores the stochastic term on price formation, 
in which it differs from Huberman and Stanzl (2004).

 Let $x^i_n, i=1,2,$ denote speculator $n$'s market order
  in his or her $i$th period, and let $y_1=x^1_1$ and $y_n= x^2_{n-1}+x^1_n$ for each $n\geqq 2$.
 Let $\xi_n$ and $u_n$ denote the market orders of the momentum traders and the controller, respectively. 
 In each trading period, active speculators, momentum traders, 
 and the controller simultaneously place their market orders,
 that is, $q_n= y_n+\xi_n+u_n\in \mathbb{R}$.
  After the trade, $(p_n, q_n)$ is public, but not $(y_n, \xi_n, u_n)$.

  Here, I make the following three assumptions about the model.
 \begin{assumption}\label{A}  
  \item[1.]  $x^i_n\in \{-1,0,1\}$ and $x^1_n+x^2_n=0$ for each 
$n\in \mathbb{N}$. 
  \item[2.] $\xi_n=\beta(p_{n-1}-p_{n-2})$ with a constant $\beta\in \mathbb{R}_{+}$ if $n\geqq 2$; and \label{a2} 
  $\xi_1\equiv 0$. 
  \item[3.]  $U_n(q_n)=\lambda q_n$ and $P_n(q_n)=\mu q_n$, with 
 constants $\lambda, \mu \in \mathbb{R}_{++}$.
 \end{assumption}
  Assumption \ref{A}-1 is a normalization of speculators' trading. 
 Each speculator's order placement is bounded, and   
  each speculator exits the market with a zero position. 
  Assumption \ref{A}-2 characterizes the behavior of momentum traders, as in De Long et al. (1990).
  Assumption \ref{A}-3 implies that the pricing rules
 are \textit{linear} and  \textit{time independent}.
 Huberman and Stanzl (2004) provide the rationale for using 
 a linear price-update function.\footnote{Huberman and Stanzl (2004) show that time-independent price-update functions
 must be \textit{quasi-linear} for market viability, 
that is,  $U(q_n)=\lambda q_n+R(q_n)$ such
 that $\mathbf{E}(R(\tilde{q}_n)\mid \mathcal{G}_n)=0$, 
for an information set $\mathcal{G}_n$.
Because there are no random variables in my model, the quasi-linearity 
  is equivalent to linearity.} 
The linearity assumption on the price-impact functions is for simplification.

\subsection{Definitions of the games}
   The set of speculators' action spaces is
  \[ 
  X=\{~ (x^1, x^2)~|~ x^i\in\{-1,0,1\},~x^1+x^2=0\}. 
  \] 
  I describe $x_n\in X$ as speculator $n$'s action. 
 Speculator $n$'s payoff is $-p_nx^1_n-p_{n+1}x^2_n=(p_{n+1}-p_{n})x^1_n$.

Let $\mathcal{H}_n$ denote the set of all possible prices and 
quantities until period $n-1$, that is,
$\mathcal{H}_1=\{\emptyset\}$ and
 $\mathcal{H}_n=(\mathbb{R}_+\times \mathbb{R})^{n-1}$.
I refer to $h_n\in \mathcal{H}_n$ as the \textit{history} until period $n-1$. 
I define the \textit{strategy} of speculator $n$ as a mapping 
$s_n:\mathcal{H}_n \to X$. 
 This study focuses on \textit{pure strategies}.  
Let $S_n$ denote the set of all possible strategies of speculator $n$
and $s\in \prod_{n\in \mathbb{N}} S_n$ denote the strategy profile of speculators. 
Let $\pi_n(s)$ denote speculator $n$'s payoff.\footnote{By definition, $\pi_1(s)=\pi_n(s_1, s_{2})$ and 
$\pi_n(s)=\pi_n(s_{n-1},s_n, s_{n+1})$ for each $n\geqq 2$. } 
Hence, $(\mathbb{N}, (S_n)_{n\in \mathbb{N}}, (\pi_n)_{n\in \mathbb{N}})$ defines a game.  
$\Gamma(h_n)$ describes the \textit{subgame} beginning at history $h_n$
and $\pi_n(s \mid h_n)$ speculator $n$'s payoff in subgame $\Gamma(h_n)$.

\begin{definition} 
The strategy profile $s$ is an \textup{NE} 
 if $\pi_n(s)\geqq \pi_n(s'_n, s_{-n})$ holds for each $n\in \mathbb{N}$
 and $s'_{n}\in S_n$, where $s_{-n}\in \prod_{k\in \mathbb{N}\setminus\{n\}}S_k$.
 The strategy profile $s$ is an \textup{SPE} if 
 $\pi_n(s\mid h_n)\geqq \pi_n(s'_n, s_{-n} \mid h_n)$ holds 
for each $n\in \mathbb{N}$, $s'_n\in S_n$, and $h_n\in \mathcal{H}_n$. 
\end{definition}

  I refer to an outcome such that $x_n=(0,0)$ for each $n\in \mathbb{N}$
as \textit{no trade}.
  The market is \textit{NE-viable (SPE-viable)} 
if, under the associated pricing rule $(\lambda,\mu)$,
any NE (SPE) outcome leads to no trade.
 Then, $(\lambda,\mu)$ is said to be a \textit{viable pricing rule}.

 \section{The Results without Controls}
   This section aims to characterize the viable sets when 
   the controller is absent.

\subsection{The benchmark model}
Since the price-impact model of this study is based on Huberman and Stanzl (2004), I employ their model as a benchmark.
They analyze a single $N$-period-lived speculator model 
without momentum traders.
The set of the speculator's action space 
is $X=\bigcup_{N=1}^\infty X_N$, where 
  \[ 
  X_N=\left\{(x_1,\ldots,x_N)\in \mathbb{R}^N \mid  \sum_{n=1}^Nx_n=0 \right\}.
  \] 
The speculator's strategy is $s:\mathcal{H}_1\to X$. 
For each $s$, there is a unique $N$ such that 
the speculator's payoff is $\hat\pi_N(s):=-\sum_{n=1}^{N} p_nx_n$.
Let $\bm{0}_N$ denote the zero vector of $\mathbb{R}^N$ and 
let $Y=\bigcup_{N=1}^\infty\{\bm{0}_N\}$. 
The market is SPE-viable if and only if the optimal strategy $s$
is such that $s:\mathcal{H}_1\to Y$.

  \begin{proposition}\label{HS}
In the benchmark model, 
 the market is SPE-viable if and only if 
$(\lambda, \mu)$ satisfies $\lambda\leqq 2\mu$. 
 \end{proposition}

 \begin{proof}
For each $s$, if $N=1$, $x_1=0$ and $\hat\pi_1(s)=0$; otherwise, 
\[
\begin{split}
\hat\pi_N(s)&=-\sum_{n=1}^N (p_0+\lambda\sum_{k=1}^{n-1}x_k+\mu x_n)x_n \\
 &=-p_0\sum_{n=1}^N x_n-\lambda\sum_{n=1}^N\left(\sum_{k=1}^{n-1}x_k \right)x_n-
\mu \sum_{n=1}^Nx_n^2\\
&=-\lambda\sum_{i<j}x_ix_j-\mu\sum_{n=1}^Nx^2_n
~~~~\left(
 \because \sum_{n=1}^Nx_n=0. \right)\\
&=(2\mu-\lambda)\left(-\frac{1}{2}\sum_{n=1}^N x_n^2\right).
~~~~\left(
 \because \sum_{i=1}^Nx_i\sum_{j=1}^Nx_j=\sum_{n=1}^Nx^2_n+
2\sum_{i<j}x_ix_j=0. \right)
\end{split}
\]
This equality implies that $2\mu-\lambda\geqq 0$ is necessary and sufficient 
to ensure SPE-viability. 
   \end{proof}

Proposition \ref{HS} implies that the permanent price impact should 
 be sufficiently small compared with the immediate price impact for 
 SPE-viability in the benchmark model.  
 Proposition \ref{HS} also defines the \textit{maximal set} of 
 viable pricing rules under  Assumption \ref{A}-3.

\begin{definition}
 The set of pricing rules $M$ is said to be the \textup{maximal set} if 
 \[
  M:=\{(\lambda, \mu)\in \mathbb{R}^2_{++} \mid \lambda \leqq 2\mu \}.
 \] 
\end{definition}

\subsection{The present model}

I return to the present model: two-period-lived speculators 
 and $\beta\geqq 0$.
I propose three results. 
The first is the following. 

 \begin{proposition} \label{p2}
  No trade NE exists if and only if 
  \[
 (\lambda,\mu)\in M_1(\beta)
 :=\{(\lambda,\mu)\in \mathbb{R}_{++} \mid  R\equiv \beta\mu^2-2\mu+\lambda\leqq 0\}.
\]
   \end{proposition}

  \begin{proof}
    Suppose that the NE strategy profile $s$ 
     realizes $x_n=(0,0)$ for each $n\in \mathbb{N}$.
It is sufficient to show that speculator 1 has no incentive to deviate 
 from $s$.
Suppose that $x_1=(1,-1)$.
Then, $p_1=p_0+\mu$ and $p_2=p_0+\lambda +\mu(\beta\mu-1)$.
The payoff is $p_2-p_1=\beta\mu^2-2\mu+\lambda$. 
I obtain the same payoff when $x_1=(-1,1)$.
Hence, NE-viability requires $\beta\mu^2-2\mu+\lambda\leqq 0$. 
  \end{proof}
  Condition $R\leqq 0$ reduces to $\lambda\leqq 2\mu$
when $\beta=0$, which implies that 
 $M_1(\beta) \subset M$ for each $\beta\geqq 0$ and $M_1(0)=M$. 
However, Proposition \ref{p2} does not ensure the uniqueness 
of the equilibrium outcomes. Hence, the market is NE-viable \textit{only if}
 $(\lambda,\mu)\in M_1(\beta)$.

 Next, I provide a condition for which we can ignore some undesirable NEs.   
  Let $s'$ denote the strategy profile, 
    which maps $x_n=(1,-1)$ for each $n\geqq  1$. 
   Then, I have $y_1=1$ and $y_n=0$ for each $n\geqq 2$. 
    From Assumptions \ref{A}-2 and \ref{A}-3, the sequence  $(y_n)^\infty_{n=1}$ generates 
    a sequence $(q_n)^{\infty}_{n=1}$ such that $q_1=y_1=1$ and 
  \[ 
     \begin{split}
    q_n =\xi_n &= \beta(p_{n-1}-p_{n-2})\\
               &= \beta((\lambda-\mu)q_{n-2}+\mu q_{n-1}) 
     \end{split}          
  \] 
      for each $n\geqq 2$.
   Then, the payoff of speculator $n$ is 
  \[ 
    \begin{split}
    p_{n+1}-p_n&=\mu q_{n+1}+ (\lambda-\mu)q_n\\
               &=\frac{q_{n+2}}{\beta}. \label{piinf}
    \end{split}
   \] 
 If  $D\equiv \beta\mu^2-4\mu+4\lambda\geqq 0$, then
 $q_n>0$ for each $n\geqq 1$ (see the Appendix).
  In this case, the market price will grow monotonically.
  Therefore, each speculator $n$ obtains a positive profit margin. 
  This observation implies that strategy profile $s'$ is an NE.
  Hence, $D<0$ is \textit{necessary} for NE-viability. 
  Let $M_2(\beta)$ denote the set of pricing rules such that
 \[
    M_2(\beta)
 :=\{(\lambda,\mu)\in \mathbb{R}_{++} \mid  D\equiv \beta\mu^2-4\mu+4\lambda< 0\}.
 \]
\begin{proposition}\label{D<0}
The market is NE-viable only if 
   \[
 (\lambda,\mu)\in M_1(\beta) \cap M_2(\beta).
\]
\end{proposition}
\begin{proof}
 See the Appendix.
\end{proof}

The converse \textit{might} not be true.
The Appendix shows that, if $D<0$,  there exist infinitely many $n$ 
such that $\xi_n<0$ on the path of $s'$. 
    In this case, there exists $n^*$ such that speculator $n^*$ is 
    \textit{better off} choosing $(0,0)$, while 
   speculators 1 to $n^*-1$ choose $(1,-1)$.
 However, speculators 1 to $n^*-1$ could gain a positive profit
even if speculators $n\geqq n^*$ choose $(0,0)$.  
Moreover, there may be other complicated strategies that bring
a positive profit for speculators. 
 
 Finally, I provide a sufficient condition that ensures 
 SPE-viability.

  \begin{proposition}\label{s}
 The market is SPE-viable if  
  \[
   (\lambda,\mu)\in  M_3(\beta)
 :=\{(\lambda,\mu)\in \mathbb{R}_{++} \mid  L\equiv \beta\mu^2-2\mu+2\lambda< 0\}.
 \]
  \end{proposition}
   \begin{proof}
 See the Appendix.
 \end{proof} 

 Because SPE-viability implies NE-viability, Proposition \ref{s}
provides a sufficient condition for NE-viability. 
In the case of no momentum traders, $\beta=0$, we have
 \[
 M_2(0)=M_3(0)\subsetneq M_1(0)=M,
 \]
 which gives the following result. 
\begin{corollary}\label{cor1} 
 We assume that $\beta=0$.
 The market 
 is NE-viable if and only if $\lambda<\mu$. 
\end{corollary}

Kyle's (1985) equilibrium pricing rule, which is an $N$-period 
trading model under asymmetric information, takes the form 
$p_n=p_{n-1}+\lambda_n q_n$. In my model, this pricing rule 
maps each price to the area $\lambda=\mu$. 
Hence, the Kyle's pricing rule does not preclude the manipulation strategy of this study.

Figure \ref{fig:3} illustrates the characterization results.  
It shows that $M_3(\beta)\subset M_1(\beta)$;
NE-viability implies $R\leqq 0$ (Proposition \ref{p2})
and this is implied by $L<0$ (Proposition \ref{s}). 
It also shows that neither $M_1(\beta)\subset M_2(\beta)$ 
nor $M_2(\beta)\subset M_1(\beta)$;
Conditions $R\leqq 0$ and $D<0$ are both necessary for NE-viability
 (Proposition \ref{D<0}), but there is no inclusive relationship.

The viable pricing rule of the present model
 is more restrictive than that of the benchmark model,
because the present model allows 
speculators to do ``support buying'' (i.e., 
 the next speculator's purchase cancels out 
   a negative impact on price by the current speculator's sale).
The chain of such behavior is profitable to speculators.

      \begin{figure}[htbp]
    \centering
   \input{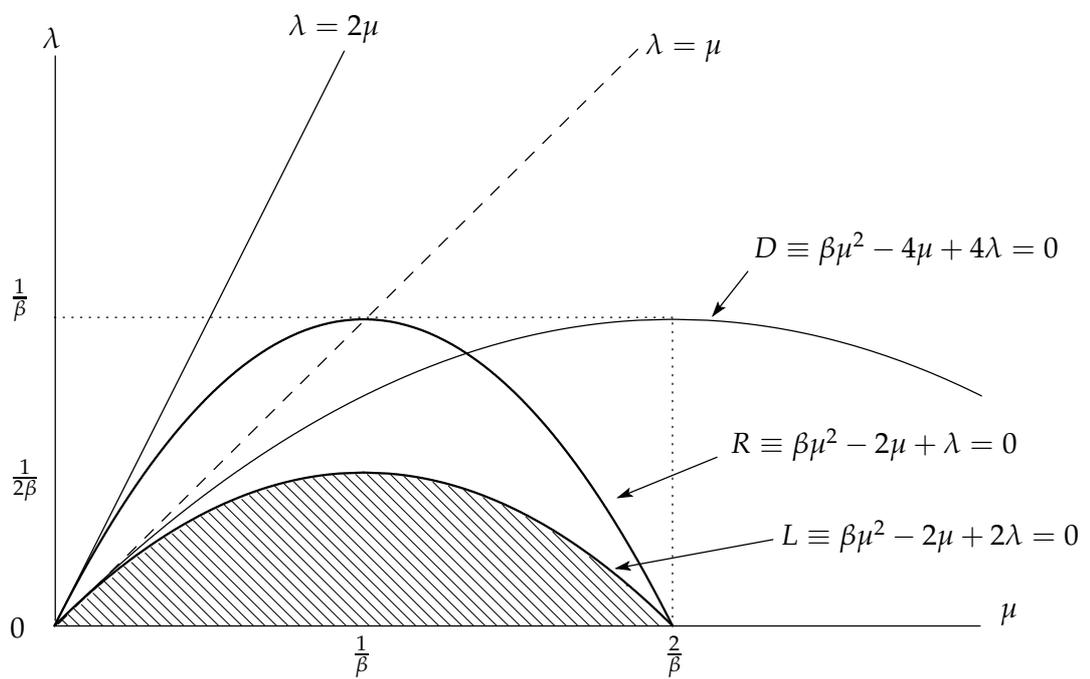}
   \vspace*{10pt}\\
   \caption{
The shaded area exhibits $M_3(\beta)$.
This figure shows that  $M_3(\beta)\subset M_1(\beta)$ but
there is no inclusive relationship between $M_1(\beta)$ and $M_2(\beta)$.}
  \label{fig:3}
   \end{figure}%

\section{The Results with Controls}
   This section aims to characterize the sets of viable pricing rules
    when  the controller is present.

Let $u_n$ denote the controller's market order placed in period $n$. 
 I refer to the sequence $u=(u_n)^\infty_{n=1}$  as  a \textit{control}. 
    Let $u^{k}$ denote a control such that $u_n=0$ for each $n<k$.

 Suppose that each speculator $n$ chooses $x_n=(1,-1)$.
Eq.\,(\ref{single}) with Assumption \ref{A} yields
\begin{equation}\label{price}
   p_{n+1} =p_n+(\lambda-\mu)q_n+\mu q_{n+1} 
\end{equation}
for each $n\geqq 1$.
Because $q_{n+1}=\xi_{n+1}+u_{n+1}$, Assumption \ref{A}
implies
\begin{equation}\label{quantity}
    q_{n+1} =\beta\mu q_n+\beta(\lambda-\mu) q_{n-1}+ u_{n+1}
    \end{equation} 
for each $n\geqq 1$. 
 I describe Eqs.\,(\ref{price}) and (\ref{quantity})
 as a difference equation system:
  \begin{equation}
\bm{z}_{n+1}    =
    A \bm{z}_n
      + 
B
     \bm{u}_n,              \label{sys}
  \end{equation}
where
\[ 
 A=
   \begin{pmatrix}
       1 & \lambda-\mu+\beta\mu^2 & \beta(\lambda-\mu)\mu \\
       0 & \beta\mu & \beta(\lambda-\mu) \\
       0 & 1    & 0    \\     
   \end{pmatrix}, 
   \hspace*{20pt}
  B= 
   \begin{pmatrix}
      \mu \\
      1   \\
      0
   \end{pmatrix} \label{AB}
 \] 
 and
 \begin{equation}\label{u}
 \bm{z}_n=
  \begin{pmatrix}
       p_{n} \\
       q_{n} \\
       q_{n-1}
     \end{pmatrix},
 \hspace*{20pt}
 \bm{u}_n=u_{n+1}.
      \end{equation}
 I refer to $\bm{z}_n$ as a \textit{state} in period $n$.

Here, the definition of controllability is presented. 
\begin{definition} \textup{(Elaydi (2005), p.~432.)}
 System (\ref{sys}) is said to be \textup{controllable} if 
 for each $k\in \mathbb{N}$, for each initial state $\bm{z}_{k-1}$, and for each
  final state  $\bm{z}^*$, there exists a finite number $N>k-1$ and 
  a control $u^k$, $k\leqq N$, such that $\bm{z}_N=\bm{z}^*$.
\end{definition} 
  If System (\ref{sys}) is controllable, the controller can set a 
 market condition to a final state within finite periods only by 
 placing market orders.
  The \textit{market control succeeds} if System (\ref{sys}) 
 is  controllable. 
  Based on discrete control theory,
   System (\ref{sys}) is controllable if and only if  
    the matrix $W= [B,~ AB,~A^2B] $  has a full row rank.\footnote{
See Theorem 10.4 of Elaydi (2005, p.~433),
which is presented in the Appendix for convenience. 
} 
  In the present model, 
    \[ 
       W=
      \begin{pmatrix}
        \mu & \mu(\beta\mu)+\lambda & \mu(\beta\mu)^2-(\mu-2\lambda)\beta\mu+\lambda \\
        1   & \beta\mu            & (\beta\mu)^2-\beta\mu+\beta\lambda \\
        0   & 1               & \beta\mu
      \end{pmatrix}
   \] 
and $\det(W)=\lambda$. 
 I summarize the above as a proposition. 
 \begin{proposition}\label{con}
   Market control succeeds if and only if $\lambda>0$. 
 \end{proposition}

   Next, I require that any control \textit{stabilizes} the market in 
the sense that,   for each $k$, $u^k$ must asymptotically satisfy
 $(p_n,q_n)\to (p_{k-1},q_{k-1})$.
 Consider System (\ref{sys}) with a linear feedback 
 $\bm{u}_n=-S\bm{z}_n$,
 where $S=(\sigma_1,\sigma_2,\sigma_3)$ is a real $(1\times 3)$
 matrix.
 Then,
 \[ 
   \begin{split}
   \bm{z}_{n+1}&=A\bm{z}_n+B\bm{u}_n \\
               &=(A-BS)\bm{z}_n.  \label{A-BK}
   \end{split}             
 \] 
  System (\ref{sys}) is \textit{stabilizable} if, for each $k$, 
there exists matrix $S$ such that control $\bm{u}^k=-S\bm{z}_k$
 achieves $\displaystyle\lim_{k\to\infty}\bm{z}_k=\bm{z}_{k-1}$.
 \textit{Market stabilization succeeds}  
  if System (\ref{sys}) is stabilizable. 
 I define $\bm{z}^*=\bm{z}_{k-1}$ and
 $\bm{y}_k=\bm{z}_k-\bm{z}^*$ to describe  
 System (\ref{sys}) as follows:  
 \[
 \bm{y}_{k+1}=A(\bm{y}_k+\bm{z}^*)+B\bm{u}_k-\bm{z}^*
 =\bm{z}_{k+1}-\bm{z}^*.\]
Stabilization to an initial state $\bm{z}^*$
is equivalent to stabilization to $\bm{0}$. 
Hence, the goal of stabilization is 
 $\bm{z}_k \to\bm{0}$, without loss of generality.
 The next proposition is useful for stabilization.

\begin{proposition}\label{E}\textup{(Elaydi (2005), Theorem 10.19.)}
  Let $\Phi  =\{\varphi _1,\varphi _2,\varphi _3\}$ be an arbitrary 
 set of  3 complex  numbers such that $\overline{\Phi }=
 \{\bar{\varphi }_1,\bar{\varphi }_2,\bar{\varphi }_3\}=\Phi $.
  Then, System (\ref{sys}) is controllable if and only if there exists
  a matrix $S$ such that the eigenvalues of $A-BS$ are the set $\Phi $. 
 \end{proposition}
Thus, System (\ref{sys}) is \textit{stabilizable} if the eigenvalues 
of $A-BS$ lie inside the unit disk. 
 Applying Proposition \ref{E} to System (\ref{sys}), I obtain the 
following result.
 \begin{proposition}\label{stab}
 Market stabilization succeeds if $\lambda>0$. 
 \end{proposition}
  \begin{proof}
 Without loss of generality, I identify $\bm{z}_{k-1}=\bm{0}$. 
 Let $\Phi =\{\varphi_1,\varphi _2,\varphi _3\}$ denote 
 the set of complex numbers such that $|\varphi_i|<1$ for each $i=1,2,3$.
 I define $\delta_1=-(\varphi_1+\varphi_2+\varphi_3)$,
  $\delta_2=\varphi_1\varphi_2+\varphi_2\varphi_3+
 \varphi_3\varphi_1$, and $\delta_3=-\varphi_1\varphi_2\varphi_3$.
The characteristic polynomial of $A-BS$ is
\[ 
  \begin{split}
 \textup{det}\left(A-BS-\varphi  I\right)  =
  \begin{vmatrix}
    1-\mu\sigma_1-\varphi  & \lambda-\mu+\beta\mu^2-\mu \sigma_2 & \beta(\lambda-\mu)\mu-\mu \sigma_3 \\
    -\sigma_1 & \beta\mu-\sigma_2-\varphi  & \beta(\lambda-\mu)-\sigma_3 \\
    0 & 1 & -\varphi 
  \end{vmatrix} 
    = 0, \label{cha}
  \end{split}
 \] 
 which is equivalent to
\[ 
\begin{split}
  \varphi ^3+  (\mu \sigma_1+\sigma_2-\beta\mu-1)\varphi ^2 
 & + ((\lambda-\mu) \sigma_1-\sigma_2
 +\sigma_3-\beta(\lambda-\mu)+\beta\mu)\varphi\\
& -\sigma_3+\beta(\lambda-\mu)=0.
\end{split}
\] 
 Comparing the coefficients with the roots, I obtain
 \begin{equation}\label{deltaone}
   \begin{split}
     \mu \sigma_1+\sigma_2-\beta\mu-1=\delta_1 &\\
     (\lambda-\mu) \sigma_1-\sigma_2 +\sigma_3-\beta(\lambda-\mu)+\beta\mu=\delta_2& \\
     -\sigma_3+\beta(\lambda-\mu)=\delta_3.
   \end{split}
 \end{equation}
These equations yield 
\begin{equation}\label{S}
  \begin{split}
    S&=(\sigma_1,\sigma_2,\sigma_3) \\
 &= \left( \frac{1+\delta_1+\delta_2+\delta_3}{\lambda},~ 
\frac{-\mu(1+\delta_1+\delta_2+\delta_3)}{\lambda}+\delta_1+\beta\mu+1,~ 
-\delta_3+\beta(\lambda-\mu)\right).
 \end{split}
  \end{equation}
 Because $|\varphi_i|<1$ for each $i$, the controller makes the state
 $(p_n,q_n,q_{n-1})$ converge to $(0,0,0)$ by placing the orders
 followed by $u^k=-S\bm{z}_k$. 
 \end{proof}

 Let $\mathcal{U}$ denote the set of matrixes, $S$ in Eq.\,(\ref{S}),
 such that the maximum absolute value of eigenvalues of 
matrix $A-BS$ is less than 1. 
 \begin{definition}
A \textup{quick response control} is a control 
 such that if there exists $S=(\sigma_1,\sigma_2,\sigma_3)\in \mathcal{U }$  such that  for each $k\geqq 1$,
\begin{itemize}
 \item $\forall l\leqq k~~q_l=0 \Rightarrow u_{k+1}=0$.
\item $q_k\neq 0 \Rightarrow 
 \forall n\geqq k+1~~u_{n}=-S\bm{z}_n=-\sigma_1p_{n-1}-\sigma_2q_{n-1}-\sigma_3q_{n-2}$.
\end{itemize}
\end{definition}
 In other words, a quick response control is a control such that (i) it achieves market stabilization 
and (ii) the controller enters the market
 \textit{only if}
 the speculator trades; otherwise, the controller never enters 
  the market (i.e., no market intervention occurs).
 The market is said to be \textit{SPE-viable under quick response controls} 
if, for each $(\lambda,\mu)\in \mathbb{R}^2_{++}$, 
there exists a quick response control such that no trade is a
unique SPE outcome.

 \begin{theorem} \label{positive}
   The market is SPE-viable under quick response controls 
if and only if $(\lambda,\mu)\in M$.
  \end{theorem} 
 \begin{proof}
I assume that $p_0=0$ for simplicity. 

Suppose that $(\lambda,\mu)\in M$. 
I first show that there is a quick response control 
that makes the market SPE-viable. 
    Let  $\Phi =\{1/3, 1/3, 1/3\}$.
Then, by Eq.\,(\ref{S}), I have 
 \[ 
    u_{n+1}=-\frac{1}{\lambda}\frac{8}{27}p_n
    +\frac{\mu}{\lambda}\left(\frac{8}{27}-\beta\lambda\right)q_n 
    -\left(\frac{1}{27}+\beta(\lambda-\mu) \right)q_{n-1}.
  \]
Let $\Gamma(x)$ denote a subgame such that a speculator $n$ chooses $x_n=(x,-x)$ and each speculator $i<n$ chooses $x_i=(0,0)$. 
  Suppose that speculator 1 implements $x_1=(1,-1)$. 
 Then,
  $u_2=-\beta\mu$ and $q_2=-1+{x}^1_2$.
 These lead to
 \[ 
    \begin{split}
       p_2 &=\lambda-\mu+\mu {x}^1_2 \\
      \xi_3&=\beta(\lambda-2\mu+\mu{x}^1_2)\\
       u_3 &=-\frac{1}{3}-\xi_3\\
        p_3 &=\mu \left( {x}^1_3-\frac{1}{3} \right) +(\lambda-\mu){x}^1_2.\\
       \end{split}
  \] 
If $x^1_2\neq 0$, a simple calculation yields
 \[
 \begin{split}
 \pi_2(s \mid 1)&=
\left(
(\lambda-2\mu){x}^1_2+\mu{x}^1_3-\lambda+\frac{2\mu}{3}
\right)x^1_2\\
&=
  \begin{cases}
     \mu\left(x^1_3-\frac{4}{3}\right) <0 & \text{if $x^1_2=1$} \\
  \mu\left(-x^1_3-\frac{8}{3}\right)+2\lambda& \text{if $x^1_2=-1$.}
  \end{cases}
  \end{split}
 \]
Because $\pi_2(s\mid 1)=0$ if $x_2=(0,0)$, any SPE makes 
 $x_2\neq (1,-1)$. 
  Then, $\pi_1(s\mid 1)=\lambda-2\mu+\mu{x}^1_2$ with 
${x}^1_2\in \{-1,0\}$.
Because $\lambda\leqq 2\mu$, $\pi_1(s\mid 1)\leqq 0$.  
The argument is symmetric in the case of $x_1=(-1,1)$. %
Hence, it is optimal for each speculator $n$ to choose $x_n=(0,0)$ in $\Gamma(0)$.

Next, I show that the market is not SPE-viable under any quick 
response control for some $(\lambda,\mu)\not\in M$. 
Let $\lambda>2\mu$. 
Consider a subgame in which only speculator 1 chooses $x_1=(1,-1)$
 and other speculators choose $(0,0)$. 
In this case, by Eq.\,(\ref{deltaone}), $u_2=-\sigma_1p_1-\sigma_2q_1
-\sigma_3q_0=-(\delta_1+\beta\mu+1)$. 
Because $(\sigma_1,\sigma_2,\sigma_3)\in \mathcal{U}$, 
we must have $\delta_1<3$. 
To achieve SPE-viability, we must have
\begin{equation}\label{fu}
p_2-p_1=(\lambda-2\mu)-\mu(1+\delta_1)\leqq 0.
\end{equation}
Eq.\,(\ref{fu}) is equivalent to 
\[
\frac{\lambda-2\mu}{\mu}-1\leqq \delta_1<3,
\]
which is impossible for a sufficiently small $\mu>0$.
   \end{proof}

 \section{Concluding remarks}
 This study models a trading system as a linear system and introduces 
a control to the system. The control affirms the prevention of
 ``momentum-ignition'' price manipulation (MIPM). 
The controllability result of the linear model 
 is applicable to other financial policies, such as inflation targeting,
 the stabilization of price bubbles, and the prevention of herding during 
 a market crash. 
The implication of this study is that it is important for the controllability 
of the market to check whether the observed price impacts in practice 
are linear or not .

 The main finding of this study is that Huberman and Stanzl's (2004) 
 benchmark result is achievable with an appropriate control, while it is
 never achievable without controls. 
The market intervention of the model, which is identical to
a control, never destabilizes the market
given a strategy profile of speculators.

The main assumption of this study (Assumption \ref{A}) is 
important to  the affirmative result on the prevention of MIPM.  
For the restriction of speculators' action spaces, the model suggests 
that the result is valid under a more general 
class of strategies. 
I briefly state my conjecture. For a strategy profile $s$ that generates a 
sequence of active speculators' aggregate quantity, $(y_n)$, 
the controller can set a control $(u'_n)$ such that 
$\bm{u}_n=u'_{n+1}-y_{n+1}=u_{n+1}$ in Eq.\,(\ref{u}). 
Hence, the result still affirms the prevention of MIPM. 
Meanwhile, the simplicity of momentum traders' behavior and time-independent pricing rules are controversial.  

Finally, my model has certain limitations: there are no stochastic terms, 
 long-lived speculators, or budget constraints of the controller. 
In particular, we should consider the controller's budget constraints
 more seriously. 
Since the controller's purchase (or sale) continues infinitely for market
 stability, it is important to show that the required budget is bounded. 
These topics are left for future research.

 \appendix
\section*{Appendix}

 \appendix
  \renewcommand{\theequation}{A.\arabic{equation}}

 \section{Second-order linear difference equation}
 A difference equation 
 appearing in the model is defined by 
  \begin{equation} 
    q_{n+2}=Kq_{n+1}+Jq_{n}, \label{diff}
  \end{equation}  
 where $K,J\in \mathbb{R}$ and $n\in \mathbb{N}$. 
The \textit{characteristic equation}
of Eq.\,(\ref{diff}) is
 \begin{equation}
   r^2-Kr-J=0. \label{char}
 \end{equation} 

Suppose that  Eq.\,(\ref{char}) has distinct characteristic roots, say $r_1$ and  $r_2$. 
Then, a sequence $(q_n)$ 
described as $q_n=c_1r_1^n+c_2r_2^n$ 
   is the solution of  Eq.\,(\ref{diff}) because
  \begin{equation}\label{yarikata}
   \begin{split}
            &q_{n+2}-Kq_{n+1}-Jq_n \\
            &= c_1r_1^{n+2}+c_2r_2^{n+2}
              -K(c_1r_1^{n+1}+c_2r_2^{n+1})
              -J(c_1r_1^n+c_2r_2^n) \\
            &= c_1r_1^n(r_1^2-Kr_1-J)
               +c_2r_2^n(r_2^2-Kr_2-J)\\
            &= 0.        
   \end{split}
 \end{equation}
  This solution is uniquely determined by an initial value of 
 Eq.\,(\ref{diff}).
 If $(q_0, q_1)=(0,1)$, then $c_1$ and $c_2$ are determined uniquely as
 \[ 
    \begin{pmatrix}
      1         & 1 \\
      r_1 & r_2
    \end{pmatrix}
     \begin{pmatrix}
      c_1       \\
      c_2      \\
    \end{pmatrix}
    =
     \begin{pmatrix}
      0 \\
      1
    \end{pmatrix}.
  \] 
  
  If  the roots of Eq.\,(\ref{char}) are real numbers, then 
  \begin{equation}\label{r}
    r_1=\frac{K+\sqrt{D}}{2},~~~ r_2=\frac{K-\sqrt{D}}{2},
  \end{equation}
    where  $D=K^2+4J>0$.
    Hence, $c_1=1/\sqrt{D}$, $c_2=-1/\sqrt{D}$,
  and 
  \begin{equation}\label{r2}
  q_n=\frac{1}{\sqrt{D}}(r^n_1-r^n_2).
  \end{equation} 

If the roots of Eq.\,(\ref{char}) are complex numbers, then
  \begin{equation}\label{iphi}
    r_1=\frac{K+i\sqrt{D'}}{2},~~~ r_2=\frac{K-i\sqrt{D'}}{2},
  \end{equation}
  where $D'=-D>0$ and $i=\sqrt{-1}$.
  Then, $q_n=(r^n_1-r^n_2)/(i\sqrt{D'})$.
   Note that $D<0$ implies $J<0$. 
 In the polar form,  $r_1=(K+i\sqrt{D'})\slash{2}=\sqrt{-J}(\cos\theta+i\sin\theta)$ with some $\theta
 \in [0, 2\pi]$.\footnote{Suppose that $r_1=\alpha+i\beta$ and $r_2=\alpha-i\beta$.
 In polar coordinates, $\alpha=r\cos \theta, \beta=r\sin \theta, r=\sqrt{\alpha^2+\beta^2}$, and 
 $\theta=\tan^{-1}(\frac{\beta}{\alpha})$. }
 Using Euler's formula, $r_1=\sqrt{-J}e^{i\theta}$ and 
  $r_2=\overline{r}_1=(K-i\sqrt{D'})/2=\sqrt{-J}e^{-i\theta}$.
 Using De Moivre's theorem, $q_n$ is described as 
  \begin{equation}
    {q}_n=\frac{2}{\sqrt{D'}}\left(\sqrt{-J}\right)^n \sin(n\theta).  \label{iq}
\end{equation} 

Suppose that  Eq.\,(\ref{char}) has the same characteristic root.
 Then, $D=0$ and $r=K/2$. 
  In this case, a sequence $(q_n)$, described as $q_n=r^n(c_1+nc_2)$, is the 
  solution of Eq.\,(\ref{diff}).
 If $(q_0, q_1)=(0,1)$, then $(c_1,c_2)=(0, 2/K)$. 
 Hence, 
 \begin{equation}\label{r3}
 q_n=n\left(\frac{K}{2}\right)^{n-1}.
 \end{equation} 

\subsection{Analyses of the model}
 Suppose that  all speculators implement ${x}=(1,-1)$
 (the argument is symmetric for the case $-{x}$ if $p_0$ is sufficiently large).
 Then, the market orders are $q_0=0, q_1=1$, and 
      \begin{equation}
        q_n  =\beta\mu q_{n-1}+\beta(\lambda-\mu)q_{n-2}. \label{q}
     \end{equation}   
     for each $n\geqq 2$. 
 Let $K=\beta\mu$, $J=\beta(\lambda-\mu)$, and $D=K^2+4J$.

\begin{lemma}\label{+}
 If $(\lambda, \mu)$ satisfies $D\geqq 0$, then the sequence of the market price, $(p_n)$,
 is monotone increasing.
\end{lemma}
\begin{proof}
Suppose that $D>0$.
By Eqs.\,(\ref{r}) and (\ref{r2}), 
    \[ 
    \begin{split}
      q_n=\frac{1}{\sqrt{D}}(r^n_1-r^n_2)
      &=\frac{1}{\sqrt{D}}
\left\{ \left(\frac{K+\sqrt{D}}{2} \right)^n - 
     \left(\frac{K-\sqrt{D}}{2} \right)^n \right\}\\
     &=\frac{r^n}{\sqrt{D}}\left\{1-\left(\frac{r_2}{r_1}\right)^n\right\}, 
       \end{split}
     \] 
The assumption $D>0$ implies $\beta>0$, $K>0$, and $r_1>0$.
Then,  $|r_1|^2-|r_2|^2=K\sqrt{D}>0$, which implies 
 $|r_1|> |r_2|$ and  $q_n>0$ for each $n\in \mathbb{N}$.
   Suppose that $D=0$. 
   By Eq.\,(\ref{r3}), $q_n=n(\frac{\beta\mu}{2})^{n-1}$.
   If $\beta>0$, then $q_n>0$ for each $n\in \mathbb{N}$. 
 By Eq.\,(\ref{single}), $p_n=p_{n-1}+\lambda q_{n-1}+\mu(q_n-q_{n-1})>p_{n-1}$.
\end{proof}

For market viability, $D<0$ is necessary.
Eq.\,(\ref{iq}) implies that the sequence $(q_n)$ oscillates and there exist 
 infinitely many $n$ such that $q_n<0$.

\section{Theorem 10.\,4.~of Elaydi (2005)}
Consider the following difference equation system:
\begin{equation}\label{L}
\bm{z}_{n+1}=A\bm{z}_n+B\bm{u}_n,
\end{equation}
where $\bm{z}$ is a $k$-dimensional vector, 
$A$ is a $(k\times k)$-matrix, 
$B$ is a $(k\times m)$-matrix, and $\bm{u}$ is an 
$m$-dimensional vector, where $m\leqq k$.  
The \textit{controllability matrix} $W$ of System (\ref{L})
 is defined as the $k\times km$-matrix 
\[
W=[B,AB,A^2B,\ldots,A^{k-1}B]. 
\]
Theorem 10.\,4.~of Elaydi (2005) states that 
System (\ref{L}) is controllable if and only if 
$\textup{rank}~W=k$. 
In the present model, $k=3$ and $m=1$.

\section{Proof of Propositions}
\subsection{Proof of Proposition \ref{D<0}}

Suppose that $D\geqq 0$.
Then, Lemma \ref{+} implies that each speculator $n$ implementing $x_n=(1,-1)$
constitutes an NE, which is a contradiction.

\subsection{Proof of Proposition \ref{s}}
 
I prove this proposition by showing the first trading speculator's optimal choice is 
no trade in any SPE. 
Let $\Gamma(x)$ denote the subgame beginning at a history such that 
a speculator $n$ chooses $x_n=(x,-x)$ and each speculator $i<n$ chooses 
$x_i=(0,0)$.

 Consider an arbitrary strategy profile $s$ such that ${x}_1=(1,-1)$.
Then, speculator 1's payoff in $\Gamma(1)$ is 
\[
\pi_1(s \mid 1) = \mu(\beta\mu+x_2^1-1)+(\lambda-\mu).
\]
Speculator 2's payoff in $\Gamma(1)$ is 
 \[
  \pi_2(s \mid 1)= 
\begin{cases}
\mu \xi_3+(\lambda-\mu)\xi_2 +\mu y_3+(\lambda-\mu)y_2
&\text{if $x_2\neq (0,0)$} \\
0& \text{if $x_2=(0,0)$,} 
\end{cases}
 \] 
 where $\xi_2=\beta\mu$,
 $\xi_3=(\beta\mu)^2+(y_2-1)\beta\mu+\beta\lambda$, 
 $y_2=-1+x_2^1$, and $y_3=x_3^1-x_2^1$. 
Then, we have
\[ \begin{split}
  \mu \xi_3+(\lambda-\mu)\xi_2&=\beta\mu(\beta\mu^2+(x^1_2-3)\mu+2\lambda)\\
  \end{split}
 \] 
 and
 \[
 \begin{split}
   \mu y_3+(\lambda-\mu)y_2&=\mu(x^1_3-x_2^1)+(\lambda-\mu)(-1+x^1_2).\\
 \end{split}
 \] 
I calculate $\pi_1(s \mid 1)$ and $\pi_2(s \mid 1)$. 
\begin{itemize}
 \item If $x_2=(1,-1)$, then
 \[
\begin{split}
 \pi_1(s\mid 1)&=\beta\mu^2-\mu+\lambda \\
 \pi_2(s\mid 1)&=\beta\mu(\beta\mu^2-2\mu+2\lambda)+\mu(x^1_3-1).\\
\end{split}
\]
\item If $x_2=(0,0)$, then
 \[
\begin{split}
 \pi_1(s\mid 1)&=\beta\mu^2-2\mu+\lambda \\
 \pi_2(s\mid 1)&=0.
\end{split}
\]
\item If $x_2=(-1,1)$, then
 \[
\begin{split}
 \pi_1(s\mid 1)&=\beta\mu^2-3\mu+\lambda \\
 \pi_2(s\mid 1)&=\beta\mu(\beta\mu^2-4\mu+2\lambda)+\mu(x^1_3+1)
-2(\lambda-\mu).\\
\end{split}
\]
\end{itemize}
I further calculate $\pi_2(s \mid 1)$.
\begin{itemize}
 \item If $x_3=(1,-1)$, then
 \begin{equation}\label{3_1}
\pi_2(s\mid 1)=
\begin{cases}
  \beta\mu(\beta\mu^2-2\mu+2\lambda) & \text{if $x_2=(1,-1)$}\\
0 & \text{if $x_2=(0,0)$}\\
 \beta\mu(\beta\mu^2-4\mu+2\lambda)+4\mu
-2\lambda & \text{if $x_2=(-1,1)$.}
\end{cases}
\end{equation}
\item If $x_3=(0,0)$, then
 \begin{equation}\label{3_2}
\pi_2(s\mid 1)=
\begin{cases}
  \beta\mu(\beta\mu^2-2\mu+2\lambda)-\mu & \text{if $x_2=(1,-1)$}\\
0 & \text{if $x_2=(0,0)$}\\
 \beta\mu(\beta\mu^2-4\mu+2\lambda)+3\mu-2\lambda & \text{if $x_2=(-1,1)$.}
\end{cases}
\end{equation}
\item If $x_3=(-1,1)$, then
 \begin{equation}\label{3_3}
\pi_2(s\mid 1)=
\begin{cases}
  \beta\mu(\beta\mu^2-2\mu+2\lambda)-2\mu & \text{if $x_2=(1,-1)$}\\
 0 & \text{if $x_2=(0,0)$}\\
 \beta\mu(\beta\mu^2-4\mu+2\lambda)+2\mu-2\lambda & \text{if $x_2=(-1,1)$.}
\end{cases}
\end{equation}
\end{itemize}
These observations show that speculator 2 never chooses
 $x_2=(1,-1)$ if $L=\beta\mu^2-2\mu+2\lambda<0$. 
In this case, $\max \pi_1(s\mid 1)=\beta\mu^2-2\mu+\lambda<0$. 
The argument is symmetric for  the case of  $x_1=(-1,1)$:
speculator 2 never chooses $x_2=(-1,1)$ in $\Gamma(-1)$ and 
speculator 1's maximum payoff in $\Gamma(-1)$ is negative.
The optimal choice of speculator 1 is, therefore,  no trade: $x_1=(0,0)$. 
This argument is applicable recursively: if $L<0$,
each speculator $n$ in the subgames $\Gamma(1)$ and $\Gamma(-1)$
 gains a negative payoff, given speculator $n+1$'s optimal actions,
 regardless of the other speculators' strategy profiles in $\Gamma(1)$ and
 $\Gamma(-1)$. 
Any SPE must have this property. 
This argument also shows that, if  $L<0$, 
 there is an SPE in which speculator 1 chooses $x_1=(0,0)$ and 
speculator $n+1$ chooses $x_{n+1}=(0,0)$ in $\Gamma(0)$, 
for each $n\geqq 1$. $\blacksquare$

The converse of Proposition \ref{s} is not true. 
Obviously, if $\lambda=\mu=0$, market prices never change.
Hence, no trade occurs in any SPE. 
Even if $\lambda\mu\neq 0$, when $L=0$, 
the bottom row of Eq.\,(\ref{3_1}) is $-2(\beta\mu^2-2\mu+\lambda)>0$. 
Furthermore, Eqs.\,(\ref{3_2}) and (\ref{3_3}) show that speculator 2 never
chooses $x_2=(1,-1)$ in $\Gamma(1)$. 
By continuity, no trade occurs in an SPE for some $(\lambda,\mu)$
such that $L\geqq 0$.  

\subsection*{Acknowledgements}
The author acknowledges the many helpful suggestions of 
seminar participants during the preparation of the paper.
The author would like to thank Editage (www.editage.com) for English language editing.

 \section*{References}
\begin{enumerate} 
\item Almgren,\,R., and Chriss,\,N., 2000. Optimal execution of portfolio transactions. 
{J. Risk} {3}, 5--39.
\item Baltzer,\,M., Jank,\,S., and Smajlbegovic,\, E., 2019.
Who trades on momentum? 
{J. Financ. Market} {42}, 56--74.   
 \item De Long,\,J.\,B., Shleifer,\,A., Summers,\,L.\,H., and Waldmann,\,R.\,J., 1990.
  Positive feedback investment strategies and destabilizing rational speculation. 
{J. Finance} {45}, 379--395.
 \item Elaydi,\,S., 2005. \textit{An Introduction to Difference Equations}. Springer, New York.
\item Hong,\,H., and Stein,\,J.\,C., 1999. A unified theory of underreaction, momentum trading, and overreaction in asset markets.
{J. Finance} {54}, 2143--2184. 
 \item Huberman,\,G., and Stanzl,\,W., 2004. Price manipulation and quasi-arbitrage.
  {Econometrica} {72}, 1247--1275.
\item Jegadeesh,\,N., and Titman,\,S., 1993. Returns to buying winners and selling losers: implications for stock market efficiency.
{J. Finance} {48}, 65--91. 
 \item Kyle,\,A. 1985. Continuous auctions and insider trading. 
 {Econometrica} {53}, 1315--1335. 
\item Moskowitz,\,T.\,J., Ooi,\,Y.\,H., and Pedersen,\,L.\,H., 2012.
Time series momentum. {J. Financ. Econ.} {104}, 228--250. 
\item Ohashi, Y., 2018. Momentum ignition price manipulation and pricing mechanisms. 
SSRN discussion paper. https://dx.doi.org/10.2139/ssrn.3136824

\end{enumerate}

\end{document}